\documentclass[a4paper,english,10pt]{article}
\usepackage[utf8]{inputenc}
\usepackage[a4paper]{geometry}
\usepackage{hyperref,url}
\usepackage{cite,cleveref}
\usepackage{enumerate}
\usepackage[inline]{enumitem}
\usepackage{verbatim,fancyvrb}
\usepackage{amsmath,amsthm,amsfonts,amssymb,stmaryrd,mathtools}
\usepackage{scalefnt}
\usepackage{etex,etoolbox,thmtools,environ}
\usepackage{xspace}
\usepackage{tikz,graphicx,subcaption,placeins}

\theoremstyle{plain}\newtheorem{theorem}{Theorem}[section]
\theoremstyle{plain}\newtheorem{corollary}[theorem]{Corollary}
\theoremstyle{plain}
\theoremstyle{plain}
\theoremstyle{plain}\newtheorem{definition}[theorem]{Definition}
\theoremstyle{remark}
\theoremstyle{remark}

\usetikzlibrary{calc,arrows,positioning,intersections,decorations.pathreplacing,shapes,automata,petri,backgrounds}

\tikzset{
	place/.style={
		circle,
		thick,
		draw=blue!75,
		fill=blue!20,
		minimum size=6mm
	},
	transition/.style={
		rectangle,
		thick,
		fill=black,
		minimum width=8mm,
		inner ysep=2pt
	}
}

\makeatletter
\newcommand{\superimpose}[2]{%
  {\ooalign{$#1\@firstoftwo#2$\cr\hfil$#1\@secondoftwo#2$\hfil\cr}}}
\makeatother

\makeatletter
	\newcommand{\footnoteref}[1]{%
		\protected@xdef\@thefnmark{\ref{#1}}\@footnotemark%
	}
\makeatother

\makeatletter
\newcommand{\customlabel}[2]{%
	\protected@write \@auxout {}{\string \newlabel {#1}{{#2}{\thepage}{#2}{#1}{}} }%
	\hypertarget{#1}{#2}
}
\makeatother

\newcommand{\rulelabel}[2][\theqatuation]{
	{\text{\scriptsize[\customlabel{#2}{\textsc{#1}}\!\!]}}\hfill
}
\newcommand{\ruleref}[1]{\eqref{#1}}

\newcommand{\rulefrac}[3][]{%
#1\;%
\begin{array}{c}
	#2
	\\[0pt] \hline \\[-10pt]
	#3
\end{array}}

\newcommand{\defeq}{\triangleq}

\DeclareMathOperator{\dom}{dom}

\newcommand{\textcode}[1]{\textnormal{\ttfamily#1}}

\newcommand{\type}[1]{\textcode{#1}\xspace}

\newcommand{\libOTM}{\textnormal{\textsf{O\kern-1ptT\kern-.5ptM}}\xspace}
\newcommand{\libSTM}{\textnormal{\textsf{S\kern-.5ptT\kern-.5ptM}}\xspace}

\newcommand{\lthrparenthesis}{(\mspace{-3.7mu}[}
\newcommand{\rthrparenthesis}{]\mspace{-3.7mu})}
\newcommand{\ctmthread}[2][t]{\lthrparenthesis#2\rthrparenthesis_{#1}}
\newcommand{\ctmthreadtr}[4]{\lthrparenthesis{#1;#2}\rthrparenthesis_{#3,#4}}

\newcommand{\ctxhole}[1][-]{[#1]}
\newcommand{\ctxP}[2][t]{\mathbb{P}_{#1}\ctxhole[#2]}

\newcommand{\ctxT}[2][t,k]{\mathbb{T}_{#1}\ctxhole[#2]}

\newcommand{\Loc}{\textsf{Loc}}
\newcommand{\Var}{\textsf{Var}}
\newcommand{\Term}{\textsf{Term}}
\newcommand{\TrName}{\textsf{TrName}}
\newcommand{\State}{\textsf{State}}

\newcommand{\ctmSt}[2]{\langle #2; #1\rangle}

\title{A Specification of Open Transactional Memory for Haskell}

\author{
	\begin{tabular}{ccc}
	Marino Miculan\thanks{Partially supported by MIUR project 2010LHT4KM (\emph{CINA}).}&\quad& Marco  Peressotti\\
	\small\href{mailto:marino.miculan@uniud.it}{\tt marino.miculan@uniud.it}
	&&
	\small\href{mailto:marco.peressotti@uniud.it}{\tt marco.peressotti@uniud.it}
	\end{tabular}\\[8pt]
	\small	Laboratory of Models and Applications of Distributed Systems \\[-.8ex]
	\small	Department of Mathematics, Informatics, and Physics\\[-.8ex]
	\small	University of Udine, Italy\\
}
\date{}

\hypersetup{
	colorlinks=true,
	linkcolor=black,
	citecolor=black,
	filecolor=black,
	urlcolor=black,
	pdftitle={A Specification of Open Transactional Memory for Haskell},
	pdfauthor={Marino Miculan and Marco Peressotti},
}
\begin{document}

\maketitle


\begin{abstract}
	Transactional memory (TM) has emerged as a promising abstraction for concurrent programming alternative to lock-based synchronizations.  However, most TM models admit only \emph{isolated} transactions, which are not adequate in multi-threaded programming where transactions have to interact via shared data \emph{before} committing.
	In this paper, we present \emph{Open Transactional Memory} (OTM), a programming abstraction supporting \emph{safe, data-driven} interactions between \emph{composable} memory transactions.
	This is achieved by relaxing isolation between transactions, still ensuring atomicity: threads of different transactions can interact by accessing shared variables, but then their transactions have to commit together---actually, these transactions are transparently \emph{merged}.
	This model allows for \emph{loosely-coupled} interactions since transaction merging is driven only by accesses to shared data, with no need to specify participants beforehand.
	In this paper we provide a specification of the OTM in the setting of Concurrent Haskell, showing that it is a conservative extension of current STM abstraction. In particular, we provide a formal semantics, which allows us to prove that OTM satisfies the \emph{opacity} criterion.
\end{abstract}


\section{Introduction}
\label{sec:introduction}

\looseness=-1
The advent of multicore architectures has emphasized the importance of abstractions supporting correct and scalable multi-threaded programming.   In this model, threads can collaborate by interacting on data structures (such as tables, message queues, buffers, etc.) kept in shared memory. 
Traditional lock-based mechanisms (like semaphores and monitors) 
used to regulate access to these shared data
are notoriously difficult and error-prone, as they easily lead to deadlocks, race conditions and priority inversions; moreover, they are not composable and hinder parallelism, thus reducing efficiency and scalability.
\emph{Transactional memory} (TM) has emerged as a promising abstraction to replace locks \cite{moss:transactionalmemorybook,st:dc1997}.  The basic idea is to mark blocks of code as \emph{atomic}; then, execution of each block will appear either as if it was executed sequentially and instantaneously at some unique point in time, or, if aborted, as if it did not execute at all. This is obtained by means of \emph{optimistic} executions: the blocks are allowed to run concurrently, and eventually if an interference is detected a transaction is restarted and its effects are rolled back.  Thus, each transaction can be viewed in isolation as a \emph{single-threaded} computation, significantly reducing the programmer's burden. Moreover, transactions are composable and ensure absence of deadlocks and  priority inversions, automatic roll-back on exceptions, and increased concurrency. 

However, in multi-threaded programming different transactions may need to interact and exchange data \emph{before} committing. 
In this situation, transaction isolation is a severe shortcoming.  A simple example is a request-response interaction between two transactions via a shared buffer, like in a master/worker situation.  We could try to synchronize the threads accessing the buffer \type{b} by means of two semaphores \verb|c1|, \verb|c2| as follows: 
\\[1ex]
\begin{minipage}[t]{.45\textwidth}
\begin{BVerbatim}[baseline=t]
// Party1 (Master)
atomically {
  <put request in b>
  up(c1);
  <some other code; may abort>
  down(c2); // wait for answer
  <get answer from b; may abort>
}
\end{BVerbatim}
\end{minipage}\hfill%
\begin{minipage}[t]{.45\textwidth}
\begin{BVerbatim}[baseline=t]
// Party2 (Worker)
atomically {
  down(c1); // wait for data
  <get request from b>
  <compute answer; may abort>
  <put answer in b>
  up(c2); 
}
\end{BVerbatim}
\end{minipage}
\\[1ex]
Unfortunately, this solution does not work: any admissible execution requires an interleaved scheduling between the two transactions, thus violating isolation; hence, the transactions deadlock as none of them can progress.  It is important to notice that this deadlock arises because interaction occurs between threads of \emph{different} transactions; 
in fact, the solution above is perfectly fine for threads outside transactions or within the same transaction.

To overcome this limitation, in this paper we propose a programming model for \emph{safe, data-driven} interactions between memory transactions.  The key observation is that \emph{atomicity} and \emph{isolation} are two disjoint computational aspects:
\begin{itemize}[noitemsep]
	\item an \emph{atomic non-isolated} block is executed ``all-or-nothing'', but its execution can overlap others' and \emph{uncontrolled} access to shared data is allowed;
	\item a \emph{non-atomic isolated} block is executed ``as it were the only one'' (i.e., in mutual exclusion with others), but no rollback on errors is provided.
\end{itemize} 
Thus, a ``normal'' block of code is neither atomic nor isolated; a mutex block (like Java \emph{synchronized} methods) is isolated but not atomic; and a usual TM transaction is a block which is both atomic and isolated.  Our claim is that \emph{atomic non-isolated blocks can be fruitfully used for implementing safe composable interacting memory transactions}---henceforth called \emph{open transactions}.

In this model, a transaction is composed by several threads, called \emph{participants}, which can cooperate on shared data.  A transaction commits when all its participants commit, and aborts if any thread aborts.  Threads participating to different transactions can access to shared data, but when this happens the transactions are \emph{transparently merged} into a single one.  For instance, the two transactions of the synchronization example above would automatically merge becoming the same transaction, so that the two threads can synchronize and proceed.  Thus, this model relaxes the isolation requirement still guaranteeing atomicity and consistency; moreover, it allows for \emph{loosely-coupled} interactions since transaction merging is driven only by run-time accesses to shared data, without any explicit coordination among the participants beforehand.

In summary, the contributions of this paper are the following:
\begin{itemize}[noitemsep]
	\item We present \emph{Open Transactional Memory}, a transactional memory model where multi-threaded transactions can interact by non-isolated access to shared data. Consistency and atomicity are ensured by transparently \emph{merging} transactions at runtime.
	
	\item We describe this model in the context of Concurrent Haskell (Section \ref{sec:cot}).   Namely, we define two monads \type{OTM} and \type{ITM}, representing the computational aspects of atomic \emph{multi-threaded open} (i.e., non-isolated) transactions and atomic \emph{single-threaded isolated} transactions, respectively.  Using the construct \type{atomic}, programs in the \type{OTM} monad are executed ``all-or-nothing'' but without isolation; hence these transactions can merge at runtime. When needed, blocks inside transactions can be executed in isolation by using the construct \textcode{isolated}.
	Both OTM and ITM transactions are \emph{composable}, 
	and we exploit Haskell type system to forbid irreversible effects inside these two monads.\footnote{In fact, OTM model can be implemented in any programming language, provided we have some means, either static or dynamic, to forbid irreversible effects inside transactions.}
	
	\item We provide a formal operational semantics of our system (Section \ref{sec:semantics}).
	 This semantics defines clearly the behaviour also in less intuitive situations, and serves as a reference for implementations.
	Using this semantics we prove that OTM satisfies the \emph{opacity} correctness criterion for transactions  \cite{gk:ppopp08}.
\end{itemize}
Some concluding remarks and directions for future work are in Section~\ref{sec:conclusions}.

\section{Concurrency in Haskell}\label{sec:background}
Haskell was born as pure lazy functional language;
side effects are handled by means of monads
\cite{pw:popl1993}.
For instance, I/O actions have type \type{IO\;a} and can be combined 
together by the monadic bind combinator \textcode{>>=}.
Therefore, the function
\textcode{putChar :: Char -> IO ()} takes a character
and delivers an I/O action that, when performed (even multiple times),
prints the given character.
Besides external inputs/outputs, values of \type{IO}
include operations with side effects on mutable (typed) cells.
A cell holding values of type \type{a}
has type \type{IORef\;a} and may be dealt with only via the following operations:
\begin{Verbatim}[tabsize=3, xleftmargin=2ex, gobble=1]
	newIORef   :: a -> IO (IORef a)
	readIORef  :: IORef a -> IO a
	writeIORef :: IORef a -> a -> IO ()
\end{Verbatim}

Concurrent Haskell \cite{pgf:popl1996}
adds support to threads which independently
perform a given I/O action as explained by
the type of the thread creation function:
\begin{Verbatim}[tabsize=3, xleftmargin=2ex, gobble=1]
	forkIO :: IO () -> IO ThreadId
\end{Verbatim}
The main mechanism for safe thread communication and synchronisation
are \emph{MVars}. A value of type \type{MVar\;a} is mutable location
(as for \type{IORef\;a}) that is either empty or full with a value of 
type \type{a}. There are two fundamental primitives to interact
with MVars:
\begin{Verbatim}[tabsize=3, xleftmargin=2ex, gobble=1]
	takeMVar :: Mvar a -> IO a
	putMvar  :: Mvar a -> a -> IO ()
\end{Verbatim}
The first empties a full location and blocks otherwise
whereas the second fills an empty location and blocks otherwise.
Therefore, MVars can be seen as one-place channels and
the particular case of \type{MVar\;()} corresponds to binary semaphores.

We refer the reader to \cite{jones:2010awkward-squad} for an 
introduction to concurrency, I/O, exceptions, and
cross language interfacing (the ``awkward squad'' of pure, lazy, functional programming).

STM Haskell \cite{hmpm:ppopp2005} builds on Concurrent Haskell
adding \emph{transactional actions} and a transactional memory for safe 
thread communication, called \emph{transactional variables} or \emph{TVars} for short.

Transactional actions have type \type{STM\;a}
and are concatenated using \type{STM} monadic ``bind'' combinator, akin I/O actions.
A transactional action remains tentative during its execution
and (its effect) is exposed to the rest of the system by
\begin{Verbatim}[tabsize=3, xleftmargin=2ex, gobble=1]
	atomically :: STM a -> IO a
\end{Verbatim}
which takes an STM action and delivers an I/O action that,
when performed, runs the transaction guaranteeing
atomicity and isolation with respect to the rest of the
system.

Transactional variables have type \type{TVar\;a} where
\type{a} is the type of the value held and, like IOrefs,
are manipulated via the interface:
\begin{Verbatim}[tabsize=3, xleftmargin=2ex, gobble=1]
	newTVar   :: a -> STM (TVar a)
	readTVar  :: TVar a -> STM a
	writeTVar :: TVar a -> a -> STM ()
\end{Verbatim}
For instance, the following code uses monadic bind
to combine a read and write operation on a transactional
variable and define a ``transactional update'':
\begin{Verbatim}[tabsize=3, xleftmargin=2ex, gobble=1]
	modifyTVar :: TVar a -> (a -> a) -> STM ()
	modifyTVar var f = do
		x <- readTVar var
		writeOTVar var (f x) 
\end{Verbatim}
Then, \textcode{atomically\;(modifyTVar\;x\;f)} delivers 
an I/O action that applies \textcode{f} to 
the value held by \textcode{x} and updates \textcode{x}
accordingly---the two steps being executed as a single atomic isolated operation.

The primitives recalled so far cover memory interaction,
but STM allows also for \emph{composable blocking}. 
In STM Haskell, blocking translates in 
``this thread has been scheduled too early, i.e., the right conditions are not fulfilled (yet)''. The programmer can tell the scheduler about this fact by means of 
the primitive:
\begin{Verbatim}[tabsize=3, xleftmargin=2ex, gobble=1]
	retry :: STM a
\end{Verbatim}
The semantics of \textcode{retry} is to abort the transaction 
and re-run it after at least one of the transactional variables
it has read from has been updated---there is no point in
blindly restarting a transaction. 

Finally, transactions can be composed as alternatives by means of
\begin{Verbatim}[tabsize=3, xleftmargin=2ex, gobble=1]
	orElse :: STM a -> STM a -> STM a
\end{Verbatim}
which evaluates its first argument, and if this results is a \textcode{retry} the second argument is evaluated discarding any effect of the first.

\section{Composable open transactions}\label{sec:cot}

In this section we present the key ideas of the paper by gradually introducing the primitives  from the \libOTM library, summarised in Figure~\ref{fig:base-interface}.

Although the OTM model can be implemented in any language,
we consider Haskell because
its expressive type system offers a perfect environment
for studying the ideas of transactional memory.  
In \cite{hmpm:ppopp2005} this has been used to 
single out computations which can be executed in transactions, 
i.e.~terms which can perform memory effects, from those which
can perform irreversible input/output effects.  In this paper
we refine further this approach by using the type system to
separate \emph{isolated} transactions from those which can
interact, and hence merged.

\begin{figure}[t]
	\centering
	\begin{BVerbatim}[tabsize=3, gobble=2]
		data ITM a
		data OTM a
		-- henceforth, t is a placeholder for ITM or OTM --
		
		-- Sequencing, do notation ------------------------
		(>>=)  :: t a -> (a -> t b) -> t b
		return :: a -> t a
		
		-- Running isolated and atomic computations -------
		atomic   :: OTM a -> IO a
		isolated :: ITM a -> OTM a
		retry    :: ITM a
		orElse   :: ITM a -> ITM a -> ITM a
		
		-- Exceptions -------------------------------------
		throw :: Exception e => e -> t a
		catch :: Exception e => t a -> (e -> t a) -> t a
		
		-- Threading --------------------------------------
		fork :: OTM () -> OTM ThreadId
		
		-- Transactional memory ---------------------------
		data OTVar a
		newOTVar     :: a -> ITM (OTVar a)
		readOTVar    :: OTVar a -> ITM a
		writeOTVar   :: OTVar a -> a -> ITM ()
	\end{BVerbatim}
	\caption{The base interface of \libOTM.}
	\label{fig:base-interface}
\end{figure}

The key point is to separate isolation from atomicity. 
In fact, isolation is a computational aspect which can be 
\emph{added} to atomic transactions. From this perspective,
we distinguish between isolated atomic actions
and (non isolated) atomic actions. The former are values
of type \type{ITM\;a} and the latter of \type{OTM\;a}. 
Each type of actions can be sequentially composed 
(by the corresponding monadic binders) 
preserving atomicity and, for the former, isolation.

The function \textcode{isolated} takes an isolated atomic action and
delivers an atomic action whose effects are guaranteed to be executed
in isolation with respect to other actions. 
Then, \textcode{atomic} takes an atomic action and delivers an I/O action
that when performed runs a transaction whose effects are kept tentative 
until it commits. 
Tentative effects are shared among all non-isolated transactions.
Therefore, any value of type \type{STM\;a} can be seen as a value of \type{ITM\;a} for the I/O they deliver is the same:
\begin{Verbatim}[tabsize=3, xleftmargin=2ex, gobble=1]
	atomically  = atomic . isolated
\end{Verbatim}

\paragraph{Isolation}
\libOTM supports composable blocking via the primitive \textcode{retry}, 
under \libSTM slogan ``a thread that has to be blocked because it has 
been scheduled too soon''. As for \libSTM, retrying a transactional action 
actually corresponds to block the threads on some condition.
Note that 
\textcode{retry\,::\,OTM\;a} is not a primitive since it can be
defined from that of \textcode{ITM} as \textcode{isolated\;retry}.

Checks may be declared as follows:
\begin{Verbatim}[tabsize=3, xleftmargin=2ex, gobble=1]
	check :: Bool -> ITM ()
	check b = if b then return () else retry
\end{Verbatim}
although similar primitives may be implemented at the runtime
level in order to use this information in thread scheduling. 

 \libOTM provides a mechanism for safe thread
communication by means of transactional variables called \emph{OTVars}, similar to \libSTM's TVars but supporting \emph{open} transactions.
These variables are values of type \type{OTVar\;a} where \type{a} is
the type of value held.
Creating, reading and writing OTVars is done via
the interface shown in Figure~\ref{fig:base-interface}. 
All these actions are both atomic and isolated as ensured by their type.
Therefore, when it comes to actions of type \type{ITM\;a},
OTVars are basically TVars;
e.g.~\textcode{modifyTVar} from \libSTM corresponds to:
\begin{Verbatim}[tabsize=3, xleftmargin=2ex, gobble=1]
	modifyOTVar :: OTVar a -> (a -> a) -> ITM ()
	modifyOTVar var f = do
		x <- readOTVar var
		writeOTVar var (f x) 
\end{Verbatim}
From its type it is immediate to see that the update is both atomic and 
isolated. In fact, read and write operations are glued together by the 
\textcode{>>=} combinator, preserving both properties.

Likewise, invariants on transactional variables can be easily
checked by composing reads and checks as follows:
\begin{Verbatim}[tabsize=3, xleftmargin=2ex, gobble=1]
	assertOTVar :: OTVar a -> (a -> Bool) -> ITM ()
	assertOTVar var p = do
		x <- readOTVar var
		check (p x)
\end{Verbatim}

\paragraph{Blocking}
A semaphore is a counter with two fundamental operation:
\textcode{up} which increments the counter and
\textcode{down} which decrements the counter if
it is not zero and blocks otherwise.
Semaphores are implemented using \libOTM as OTVars
holding a counter:
\begin{Verbatim}[tabsize=3, xleftmargin=2ex, gobble=1]
	type Semaphore = OTVar Int
\end{Verbatim}
Then, \textcode{up} and \textcode{down} are
two trivial atomic and isolated updates,
with the latter being guarded by a pre-condition:
\begin{Verbatim}[tabsize=3, xleftmargin=2ex, gobble=1]
	up :: Semaphore -> ITM ()
	up s = modifyOTvar s (1+)
	
	down :: Semaphore -> ITM ()
	down s = do
		assertOTVar s (> 0)
		modifyOTVar s (-1+)
\end{Verbatim}

Actions can also be composed as alternatives
by means of the primitive \textcode{orElse}.
For instance, the following takes a family of semaphores
and delivers an action that decrements one of them, blocking
only if none can be decremented:
\begin{Verbatim}[tabsize=3, xleftmargin=2ex, gobble=1]
	downAny :: [Sempahore] -> ITM ()
	downAny (x:xs) = down x `orElse` downAny xs
	downAny [] = retry
\end{Verbatim}

\paragraph{Interaction}
The interchangeability of \libOTM and \libSTM ends
when isolation is dropped. In fact, \libOTM offers shared OTVars as a 
mechanism for safe \emph{transaction interaction}.
This means that non-isolated transactional actions see
the effects on shared variables of any other non-isolated
transactional action, as they are performed concurrently
on the same object. This flow of information introduces
dependencies between concurrent tentative actions tying
together their fate: an action cannot make its effects permanent, if it 
depends on informations produced by another action which fails to complete. 
\libOTM guarantees coherence of transactional actions 
in presence of interaction through shared transactional
variables.  Thus, OTVars enables 
loosely-coupled interaction right inside atomic actions 
taking the programming style of \libSTM a step further. 
For instance, communication, rendezvous, brokering,
and in general, multi-party interactions can all be atomic (non-isolated)
actions.

In order to substantiate these claims, let us see open transactions in action by implementing a synchronisation scenario as described in Section~\ref{sec:introduction}.
In this example a master process outsources part of an atomic computation to some thread chosen from a worker pool; data is exchanged via some shared variable, whose access is coordinated by a pair of semaphores. Notably, both the master and the worker can abort the computation at any time, leading the other party to abort as well. 
This can be achieved straightforwardly using \libOTM:
\\[1ex]
\hspace*{1.5ex}
\begin{BVerbatim}[tabsize=3, xleftmargin=2ex, gobble=1]
	master c1 c2 = do
		-- put request
		isolated (up c1)
		-- do something else
		isolated (down c2)
		-- get answer
\end{BVerbatim}
\hfill
\begin{BVerbatim}[tabsize=3, xleftmargin=2ex, gobble=1]
	worker c1 c2 = do
		-- do something
		isolated (down c1)
		-- get request
		-- put answer
		isolated (up c2)
\end{BVerbatim}
\\[1ex]
Both functions deliver atomic actions in
\textcode{OTM}, and hence are not isolated. We used semaphores 
for the sake of exposition but we could synchronize by means of more abstract
mechanisms, like barriers, channels or futures, 
which can be implemented using 
\libOTM. 

\paragraph{Concurrency}
Differently from \libSTM, \libOTM supports parallelism inside non-isolated transactions. We can easily fork new threads without leaving \type{OTM} but, like any effect of a transactional action, thread creation and execution remain tentative until the whole transaction commits.  Forked threads participate to their transaction and impact its life-cycle (e.g.~issuing aborts) as any other participant.
This means that before committing, all forked threads have to complete their transactional action, i.e.~terminate.
Therefore, although the whole effect delivered by the transaction has happened concurrently, forked threads  never leave a transaction alive.

Because of their transactional nature, threads forked inside
a transaction do not have compensations nor continuations 
(i.e.~I/O actions to be executed after an abort or after a commit).
Compensations are pointless since aborts revert all effects including 
thread creation. It is indeed possible to replace the primitive 
\textcode{fork} with one supporting I/O actions as continuations like
\begin{Verbatim}[tabsize=3, xleftmargin=2ex, gobble=1]
	forkCont :: OTM a -> (a -> IO ()) -> OTM ThreadID
\end{Verbatim}
In fact, this mechanism can be implemented by means of the primitives
already offered \libOTM:  since commits are synchronisation points,
the above corresponds to the parent thread
forking a thread for each continuation, after the atomic action is
successfully completed.

On the other hand, by definition isolated atomic actions have to appear as being executed in a single-threaded setting; hence  \type{ITM}, like \type{STM}, does not support thread creation. 

\section{Formal specification of \libOTM}
\label{sec:semantics}

\begin{figure}[t]
\centering
	\begin{tabular}{lrl}
		Value & $V\Coloneqq$ & 
		$r \mid \textcode{\textbackslash $x$\;->\;$M$}
		\mid \textcode{return\;$M$}
		\mid \textcode{$M$\;>>=\;$N$}
		\mid \textcode{throw\;$M$}
		\mid \textcode{catch\;$M$\;$N$}
		\mid \textcode{putChar\;$c$}
		\mid $\\&&$
			 \textcode{getChar}
		\mid \textcode{fork\;$M$}
		\mid \textcode{atomic\;$M$\;$N$}
		\mid \textcode{isolated\;$M$}
		\mid \textcode{retry}
		\mid \textcode{$M$\;`orElse`\;$N$} 
		\mid $\\&&$
			 \textcode{newOTVar\;$M$} 
		\mid \textcode{readOTVar\;$r$}
		\mid \textcode{writeOTVar\;$r$\;$M$}
		$\\
		Term & $M, N\Coloneqq$ & $x \mid V \mid M\,N \mid \dots$
	\end{tabular}
	\caption{The syntax of values and terms.}
	\label{fig:syntax}
\end{figure}

\begin{figure}[t]
	\centering
	\begin{tabular}{lrl}
		Thread & $T_t \Coloneqq$&$\ctmthread{M} \mid \ctmthreadtr{M}{N}{t}{k}$
		\\
		Thread family & $P \Coloneqq$ & $T_{t_1} \parallel \dots \parallel T_{t_n}\qquad \forall i,j\ t_i \neq t_j  $
		\\
		Expression &$\mathbb{E}\Coloneqq$&
		$\ctxhole \mid \textcode{$\mathbb{E}$ >>= $M$}$
		\\
		Plain process &$\mathbb P_t \Coloneqq$&
		$\ctmthread{\mathbb E}  \parallel P \hfill t \notin P$
		\\
		Transaction &$\mathbb T_{t,k} \Coloneqq$&
		$\ctmthreadtr{\mathbb E}{M}{t}{k}  \parallel P \hfill t \notin P$
		\\
		Any process &$\mathbb A_t \Coloneqq$&
		$\mathbb{P}_t \mid \mathbb{T}_{t,k}$	
	\end{tabular}
	\caption{Threads and evaluation contexts.}
	\label{fig:ctm-contextes}
\end{figure} 

\subsection{Syntax and abstract machine states}
We fix an Haskell-like language extended with the \libOTM primitives
of Figure~\ref{fig:base-interface}. The syntax is summarised in
Figure~\ref{fig:syntax} where the meta-variables $x$ and $r$
range over a given countable set of variables \Var\ and
of location names \Loc, respectively. We assume Haskell typing
conventions and denote the set of all well-typed terms by $\Term$.

Terms of this language are evaluated by an abstract state 
machine whose states are pairs $\ctmSt{\Sigma}{P}$ formed by:
\begin{itemize}[noitemsep]
	\item
	a \emph{thread family} (process) $P = T_{t_1} \parallel \dots \parallel T_{t_n}$,
	\item a \emph{memory} $\Sigma = \langle \Theta, \Delta ,\Psi \rangle$, where $\Theta : \Loc \rightharpoonup \Term$ is the \emph{heap} and $\Delta : \Loc \rightharpoonup \Term \times \TrName$ is the \emph{working memory};
	$\TrName$ is a set of names used to identify active transactions;
	$\Psi$ is a forest of threads identifiers.
\end{itemize}

\paragraph{Threads}
Threads are the smaller unit of execution the machine scheduler operates on; they execute \libOTM terms and do not have any private transactional memory.
A thread outside transactions is represented by $\ctmthread{M}$ where $M$ is the term being evaluated and $t$ is a unique \emph{thread identifier} (Figure~\ref{fig:ctm-contextes}). A thread inside a transaction $k$ is represented by $\ctmthreadtr{M}{N}{t}{k}$ where	$M$	is the term being evaluated inside the transaction $k$ and 
$N$ is the term being evaluated as \emph{continuation} after $k$ commits or aborts.

At any time, all thread identifiers are stored in the auxiliary structure $\Psi$, which is a forest reflecting how threads are forked: if $t'$ has been forked by $t$ while inside $k$ then $t'$ belongs to $k$ too and occurs in $\Psi$ as a child of $t$.

We shall present thread families borrowing the parallel operator $\parallel$ from process algebra (Figure~\ref{fig:ctm-contextes}). The operator is associative, commutative and defined only on threads whose thread identifiers are distinct. The notation is extended to thread families (i.e.~processes) with $\mathbf{0}$ denoting the empty family.

\paragraph{Memory}
The memory $\Sigma$ is divided in the heap $\Theta$ and in the distributed
working memory $\Delta$ (plus the auxiliary structure $\Psi$ recording thread fork hierarchy). 
As for traditional closed (ACID) transactions
(e.g.~\cite{hmpm:ppopp2005}), operations inside a transaction
are evaluated against $\Delta$ and effects are propagated to $\Theta$
only on commits.  When a thread inside a transaction $k$ accesses a
location outside $\Delta$ the location is \emph{claimed by transaction $k$} and
remains claimed until $k$ commits, aborts or restarts. Threads in $k$
can interact only with locations claimed by $k$, but active transactions can be merged to share their claimed locations.

We shall denote the set of all possible states as $\State$, and reference to each projected component of $\Sigma$ by a subscript, i.e. $\Sigma_\Theta$ for the heap and $\Sigma_\Delta$ for the working memory.
When describing  updates to the memory $\Sigma$, we adopt the convention that $\Sigma'$ has to be intended equals to $\Sigma$ except if stated otherwise, i.e. by statements like $\Sigma'_\Theta = \Sigma_\Theta[r \mapsto M]$. Finally, the completely undefined partial function $\varnothing$ denotes the empty heap and working memory.

\subsection{Operational semantics}

\begin{figure}[t]
	\begin{gather*}
		\rulefrac[{\rulelabel[Eval]{rule:admn-eval}}]{
			M \not\equiv V \quad \mathcal{V}[M] = V
		}{
			M \to V
		}
\\
		\rulefrac[{\rulelabel[BindVal]{rule:admn-bindv}}]{
		}{
			\textcode{return\;$M$\;>>=\;$N$} \to \textcode{$N$\,$M$}
		}
		\qquad
		\rulefrac[{\rulelabel[BindEx]{rule:admn-binde}}]{
			\textcode{e} \in \{\textcode{retry},\textcode{throw\;$N$}\}
		}{
			\textcode{e\;>>=\;$M$} \to \textcode{e}
		}
\\
		\rulefrac[{\rulelabel[CatchVal]{rule:admn-catchv}}]{
			\textcode{r} \in \{\textcode{retry},\textcode{return\;$N$}\}
		}{
			\textcode{r\;`catch`\;$M$} \to \textcode{r}
		}
		\quad
		\rulefrac[{\rulelabel[CatchEx]{rule:admn-catche}}]{
		}{
			\textcode{throw\;$M$\;`catch`\;$N$} \to \textcode{$N$\,$M$}
		}
	\end{gather*}
	\caption{Term reductions: $\scriptsize M \to N$.}
	\label{fig:semantics-term}
\end{figure}	

\begin{figure}[t]
	\begin{gather*}
		\rulefrac[{\rulelabel[InChar]{rule:input-char}}]{
		}{
			\ctmSt{\Sigma}{\ctxP{\textcode{getChar}}} 
			\xrightarrow{?c}
			\ctmSt{\Sigma}{\ctxP{\textcode{return\;$c$}}}
		}		
\\
		\rulefrac[{\rulelabel[OutChar]{rule:output-char}}]{
		}{
			\ctmSt{\Sigma}{\ctxP{\textcode{putChar\;$c$}}} 
			\xrightarrow{!c}
			\ctmSt{\Sigma}{\ctxP{\textcode{return\;()}}}
		}
\\
		\rulefrac[{\rulelabel[TermIO]{rule:termio}}]{
			M \to N
		}{
			\ctmSt{\Sigma}{\ctxP{M}} 
			\xrightarrow{}
			\ctmSt{\Sigma}{\ctxP{N}}
		}
\\
		\rulefrac[{\rulelabel[ForkIO]{rule:forkio}}]{
			t'\notin \mathsf{threads}{\ctxP{\textcode{fork\;$M$}}}
		}{
			\ctmSt{\Sigma}{\ctxP{\textcode{fork\;$M$}}}
			\xrightarrow{}
			\ctmSt{\Sigma}{\ctxP{\textcode{return\;$t'$}} \parallel 
				\ctmthreadtr{M}{\textcode{return}}{t'}{k}}
		}
	\end{gather*}
	\caption{\type{IO} state transitions.}
	\label{fig:semantics-io}
\end{figure}	

\begin{figure}[!ht]
	\begin{gather*}
		\rulefrac[{\rulelabel[TermT]{rule:term}}]{
			M \to N
		}{
			\ctmSt{\Sigma}{\ctxT{M}} 
			\xrightarrow{\tau}
			\ctmSt{\Sigma}{\ctxT{N}}
		}
	\\
		\rulefrac[{\rulelabel[ForkT]{rule:fork}}]{
			t'\notin \mathsf{threads}(\ctxT{\textcode{fork\;$M$}})
			\qquad
			\Sigma'_\Psi=
			\mathsf{add\_child}(t,t',\Sigma_\Psi)
		}{
			\ctmSt{\Sigma}{\ctxT{\textcode{fork\;$M$}}}
			\xrightarrow{\tau}
			\ctmSt{\Sigma'}{\ctxT{\textcode{return\;$t'$}} \parallel 
				\ctmthreadtr{M}{\textcode{return}}{t'}{k}}
		}
	\\
		\rulefrac[{\rulelabel[NewVar]{rule:newvar}}]{
			r \notin \dom(\Sigma_\Theta)\cup\dom(\Sigma_\Delta)\qquad
			\Sigma'_\Delta = \Sigma_\Delta[r\mapsto (M,k)]
		}{
			\ctmSt{\Sigma}{\ctxT{\textcode{newOTVar\;$M$}}}
			\xrightarrow{\tau}
			\ctmSt{\Sigma'}{\ctxT{\textcode{return\;$r$}}}
		}
	\\
		\rulefrac[{\rulelabel[Read1]{rule:read-miss}}]{
			r \notin \dom(\Sigma_\Delta) \quad 
			\Sigma_\Theta(r) = M \quad 
			\Sigma'_\Delta = \Sigma_\Delta[r \mapsto (M,k)]
		}{
			\ctmSt{\Sigma}{\ctxT{\textcode{readOTVar\;$r$}}}
			\xrightarrow{\tau}
			\ctmSt{\Sigma'}{\ctxT{\textcode{return\;$M$}}}
		}
	\\
		\rulefrac[{\rulelabel[Read2]{rule:read-hit}}]{
			\Sigma_\Delta(r) = (M,j) \qquad
			\Sigma'_\Delta = \Sigma_\Delta[k \mapsto j]
		}{
			\ctmSt{\Sigma}{\ctxT{\textcode{readOTVar\;$r$}}}
			\xrightarrow{\tau}
			\ctmSt{\Sigma'}{\ctxT[t,j]{\textcode{return\;$M$}}}
		}
	\\
		\rulefrac[{\rulelabel[Write1]{rule:write-miss}}]{
			r \notin \dom(\Sigma_\Delta) \qquad
			\Sigma'_\Delta = \Sigma_\Delta[r \mapsto (M,k)]
		}{
			\ctmSt{\Sigma}{\ctxT{\textcode{writeOTVar\;$r$\;$M$}}}
			\xrightarrow{\tau}
			\ctmSt{\Sigma'}{\ctxT{\textcode{return\;()}}}
		}
	\\
		\rulefrac[{\rulelabel[Write2]{rule:write-hit}}]{
			\Sigma_\Delta(r) = (N,j) \qquad
			\Sigma'_\Delta = \Sigma_\Delta[k \mapsto j][r \mapsto (M,j)]
		}{
			\ctmSt{\Sigma}{\ctxT{\textcode{writeOTVar\;$r$\;$M$}}}
			\xrightarrow{\tau}
			\ctmSt{\Sigma'}{\ctxT[t,k]{\textcode{return\;()}}[k \mapsto j]}
		}
	\\
		\rulefrac[{\rulelabel[Or1]{rule:orfirst}}]{
			\textcode{op} \in \{\textcode{throw}, \textcode{return}\}
			\qquad
			\ctmSt{\Sigma}{\ctmthreadtr{M}{\textcode{return}}{t}{k}} 
			\xrightarrow{\tau}^*
			\ctmSt{\Sigma'}{\ctmthreadtr{\textcode{op\;$N$}}{\textcode{return}}{t}{j}}
		}{
			\ctmSt{\Sigma}{\ctxT{\textcode{$M$\;`orElse`\;$M'$}}}
			\xrightarrow{\tau}
			\ctmSt{\Sigma'}{\mathbb{T}'_{t,j}[\textcode{op\;$N$}]}
		}
	\\
		\rulefrac[{\rulelabel[Or2]{rule:orsecond}}]{
			\ctmSt{\Sigma}{\ctmthreadtr{M}{\textcode{return}}{t}{k}} 
			\xrightarrow{\tau}^*
			\ctmSt{\Sigma'}{\ctmthreadtr{\textcode{retry}}{\textcode{return}}{t}{j}}
		}{
			\ctmSt{\Sigma}{\ctxT{\textcode{$M$\;`orElse`\;$M'$}}}
			\xrightarrow{\tau}
			\ctmSt{\Sigma}{\ctxT{M'}}
		}	
	\\
		\rulefrac[{\rulelabel[Isolated]{rule:isolated}}]{
			\textcode{op} \in \{\textcode{throw}, \textcode{return}\}
			\qquad 
			\ctmSt{\Sigma}{\ctmthreadtr{M}{\textcode{return}}{t}{k}} 
			\xrightarrow{\tau}^*
			\ctmSt{\Sigma'}{\ctmthreadtr{\textcode{op\;$N$}}{\textcode{return}}{t}{j}}
		}{
			\ctmSt{\Sigma}{\ctxT{\textcode{isolated\;$M$}}} 
			\xrightarrow{\tau}
			\ctmSt{\Sigma'}{\ctxT[t,j]{\textcode{op\;$N$}}}
		}
	\end{gather*}
	\caption{Transactional state transitions: $\scriptsize\ctmSt{\Sigma}{P} \xrightarrow{\tau}\ctmSt{\Sigma'}{P'}$.}
	\label{fig:semantics-tau}
\end{figure}	

\begin{figure}[!ht]
	\begin{gather*}
		\rulefrac[{\rulelabel[New]{rule:tr-new}}]{
		}{
			\ctmSt{\Sigma}{\ctmthread{\textcode{atomic\;$M$\;>>=\;$N$}}}
			\xrightarrow{new\langle k\rangle}
			\ctmSt{\Sigma}{\ctmthreadtr{M}{N}{t}{k}}
		}
	\\
		\rulefrac[{\rulelabel[Commit]{rule:tr-commit}}]{
			\Sigma'_\Theta = \mathsf{commit}(k,\Sigma) \qquad
			\Sigma'_\Delta = \mathsf{cleanup}(k,\Sigma)
		}{
			\ctmSt{\Sigma}{\ctmthreadtr{\textcode{return\;$M$}}{N}{t}{k}}
			\xrightarrow{co\langle k\rangle}
			\ctmSt{\Sigma'}{\ctmthread{\textcode{return\;$M$\;>>=\;$N$}}}
		}
	\\
		\rulefrac[{\rulelabel[Abort1]{rule:tr-abort-1}}]{
			\Sigma'_\Theta = \mathsf{leak}(k,\Sigma) \qquad
			\Sigma'_\Delta = \mathsf{cleanup}(k,\Sigma) \\
			\Sigma'_\Psi = \mathsf{remove}(r,\Sigma_\Psi) \qquad
			r = \mathsf{root}(t,\Sigma_\Psi)
		}{
			\ctmSt{\Sigma}{\ctmthreadtr{\textcode{throw\;$M$}}{N}{t}{k}}
			\xrightarrow{ab\langle k, t, M\rangle}
			\ctmSt{\Sigma'}{\ctmthread{\textcode{throw\;$M$\;>>=\;$N$}}}
		}
	\\
		\rulefrac[{\rulelabel[Abort2]{rule:tr-abort-2}}]{
			\Sigma'_\Theta = \mathsf{leak}(k,\Sigma) \qquad
			\Sigma'_\Delta = \mathsf{cleanup}(k,\Sigma) \\
			\Sigma'_\Psi = \mathsf{remove}(r,\Sigma_\Psi) \qquad
			r = \mathsf{root}(t,\Sigma_\Psi) \quad
			r = \mathsf{root}(t',\Sigma_\Psi)
		}{
			\ctmSt{\Sigma}{\ctmthreadtr{M'}{N}{t'}{k}}
			\xrightarrow{\overline{ab}\langle k, t, M\rangle}
			\ctmSt{\Sigma'}{\ctmthread[t']{\textcode{throw\;$M$\;>>=\;$N$}}}
		}
	\\
		\rulefrac[{\rulelabel[Abort3]{rule:tr-abort-3}}]{
			\Sigma'_\Theta = \mathsf{leak}(k,\Sigma) \qquad
			\Sigma'_\Delta = \mathsf{cleanup}(k,\Sigma) 
			\\
			\Sigma'_\Psi = \mathsf{remove}(r,\Sigma_\Psi) \qquad
			r = \mathsf{root}(t,\Sigma_\Psi) \qquad
			r \neq \mathsf{root}(t',\Sigma_\Psi)
		}{
			\ctmSt{\Sigma}{\ctmthreadtr{M'}{N}{t'}{k}}
			\xrightarrow{\overline{ab}\langle k, t, M\rangle}
			\ctmSt{\Sigma'}{\ctmthread[t']{\textcode{retry}}}
		}
	\\
		\rulefrac[{\rulelabel[MCastAb]{rule:tr-multicast-abort}}]{
			\ctmSt{\Sigma}{P} \xrightarrow{ab\langle k, t, M\rangle} \ctmSt{\Sigma'}{P'}
			\quad
			\ctmSt{\Sigma}{Q} \xrightarrow{\overline{ab}\langle k, t, M\rangle} \ctmSt{\Sigma'}{Q'}
			\hfill
		}{
			\ctmSt{\Sigma}{P\parallel Q} \xrightarrow{ab\langle k, t, M\rangle} \ctmSt{\Sigma'}{P'\parallel Q'}
		}
	\\
		\rulefrac[{\rulelabel[MCastCo]{rule:tr-multicast-commit}}]{
			\ctmSt{\Sigma}{P} \xrightarrow{co\langle k\rangle} \ctmSt{\Sigma'}{P'}
			\quad
			\ctmSt{\Sigma}{Q} \xrightarrow{co\langle k\rangle} \ctmSt{\Sigma'}{Q'}
			\hfill
		}{
			\ctmSt{\Sigma}{P \parallel Q}  \xrightarrow{co\langle k\rangle} \ctmSt{\Sigma'}{P' \parallel Q'}	
		}
	\\
		\rulefrac[{\rulelabel[MCastGroup]{rule:tr-multicast-context}}]{
			\ctmSt{\Sigma}{P} \xrightarrow{\beta} \ctmSt{\Sigma'}{P'}
			\quad
			\beta \neq \tau 
			\quad 
			\mathsf{transaction}(\beta) \notin  \mathsf{transactions}(Q)
			\hfill
		}{
			\ctmSt{\Sigma}{P \parallel Q} \xrightarrow{\beta} \ctmSt{\Sigma'}{P' \parallel Q}	
		}
	\end{gather*}
	\caption{Transaction management transitions: $\scriptsize\ctmSt{\Sigma}{P} \xrightarrow{\beta}\ctmSt{\Sigma'}{P'}$.}
	\label{fig:semantics-trs-mgr}
\end{figure}	

\begin{figure}[!ht]
	\begin{align*}
		\mathsf{threads}(T_{t_1} \parallel \dots \parallel T_{t_n}) &\defeq \{t_1, \dots t_n\}
\\
		\mathsf{transaction}(\beta) &\defeq k \text{ for } 
		\beta\in\{new\langle k\rangle, co\langle k\rangle, ab\langle k, t, M\rangle, \overline{ab}\langle k, t, M\rangle\}
\\
		(\Delta[k \mapsto j])(r) &\defeq
		\begin{cases}
			\Delta(r) &\!\text{if}\ \Delta(r)= (M,l),l\neq k\\
			(M,j) &\!\text{if}\ \Delta(r) = (M, k)
		\end{cases}
\\
		\mathsf{transactions}(P) &\defeq 
		\begin{cases}
			\mathsf{transactions}(P_1) \cup \mathsf{transactions}(P_2) & 
				\!\text{if } P = P_1 \parallel P_2\\
			\{k\}&\!\text{if } P = \ctmthreadtr{M}{N}{t}{k}\\
			\emptyset &\!\text{otherwise}
		\end{cases}
\\
		P[k \mapsto j] &\defeq 
		\begin{cases}
			P_1[k \mapsto j] \parallel P_2[k \mapsto j] &\!\text{if } P = P_1 \parallel P_2\\
			\ctmthreadtr{M}{N}{t}{j} &\!\text{if } P = \ctmthreadtr{M}{N}{t}{k}\\
			P &\!\text{otherwise}
		\end{cases}
\\
		\Theta[r \mapsto M](s) &\defeq
		\begin{cases}
			M &\!\text{if } r = s\\
			\Theta(s) &\!\text{otherwise}
		\end{cases}
\\
		\Delta[r \mapsto (M,k)](s) &\defeq
		\begin{cases}
			(M,k) &\!\text{if } r = s\\
			\Delta(s) &\!\text{otherwise}
		\end{cases}
\\
		\mathsf{cleanup}(k,\Sigma)(r) &\defeq 
		\begin{cases}
			\perp &\!\!\text{if } \Sigma_\Delta(r) = (M,k)\\
			\Sigma_\Delta(r) &\!\!\text{otherwise}
		\end{cases}
\\
		\mathsf{commit}(k,\Sigma)(r) &\defeq 
		\begin{cases}
			M &\!\text{if } \Sigma_\Delta(r) = (M,k)\\
			\Sigma_\Theta(r) & \!\text{otherwise}
		\end{cases}
\\
		\mathsf{leak}(k,\Sigma)(r) &\defeq 
			M \!\text{ if } \Sigma_\Theta(r) = M \text{ or } \Sigma_\Theta(r) = {\perp} \text{ and } \Sigma_\Delta(r) = (M,k)
	\end{align*}
	\caption{Auxiliary functions.}
	\label{fig:semantics-aux}
\end{figure}

The dynamics of the machine is defined by two transition relations presented in \Cref{fig:semantics-term,fig:semantics-io,fig:semantics-tau,fig:semantics-trs-mgr}. 
The first relation $M \to N$ is defined on terms only and models pure
computations (Figure~\ref{fig:semantics-term}). Rule \ruleref{rule:admn-eval} allows
a term $M$ that is not a value to be evaluated by means of an
auxiliary (partial) function $\mathcal{V}[M]$ yielding the value $V$; the other rules define the semantics of the monadic
\type{bind} and exception handling in a standard way.
We remark the symmetry between \type{bind} and \type{catch} and how \textcode{retry} is treated as an exception by \ruleref{rule:admn-binde} and as a result value by \ruleref{rule:admn-catchv}.

Relation $\to$ can be thought as accessory to the second relation $\ctmSt{\Sigma}{P} \xrightarrow{\beta} \ctmSt{\Sigma'}{P'}$, which describes state transitions. 
Since several rules can apply to a given state according to different evaluation contexts as per Figure~\ref{fig:ctm-contextes}, this relation is non-deterministic; this models the fact that the scheduler can choose which thread to execute next among various possibilities.
Labels $\beta$ describe the kind of transition, and are defined as follows:
\[
\beta ::= \tau
	\mid new\langle k\rangle 
	\mid co\langle k\rangle
	\mid ab\langle k, t, M\rangle
	\mid \overline{ab}\langle k, t, M\rangle
	\qquad
	\text{for  $k\in\TrName$, $M\in\Term$}
\]

Transitions labelled by $\tau$ represent \emph{internal} steps of transitions,
i.e., steps which do not need a coordination among transactions:
reduction of pure terms, thread creation and memory operations. 
These transitions are defined by the rules in Figure~\ref{fig:semantics-tau}.
Reading a location falls into two cases: rule \ruleref{rule:read-miss}
models the reading of an unclaimed location and its memory effect is
to record the claim in $\Delta$, while rule \ruleref{rule:read-hit}
models the reading of a claimed location and its effect is
to merge the transactions of the current thread with that claiming the
location. Writes behave similarly.
Rules \ruleref{rule:orfirst} and \ruleref{rule:orsecond}
describe the semantics of alternative sub-transactions: if the
first one \textcode{retry}-es the second is executed discarding any effect of the first.
Rule \ruleref{rule:fork} spawns a new thread for the current transaction; a term \type{fork M} can appear inside \textcode{atomic}, 
thus allowing multi-threaded open transactions, but its use inside \type{isolated}
is prevented by the type system and by the shape of \ruleref{rule:isolated} as well.

The remaining labels describe state transitions concerning the life-cycle of
transactions: creation, commit, abort, and restart (Figure~\ref{fig:semantics-trs-mgr}). These operations require a
coordination among threads; for instance, an abort from a thread has to be
propagated to every thread participating to the same
transaction.  This is captured in the semantics by labelling the transition
with the operation and the name of the transaction involved;
this information is used by the derivation rules to force synchronisation of
all participants of that transaction. 
To illustrate this mechanism, we describe the commit of a transaction $k$,
namely $\ctmSt{\Sigma}{P} \xrightarrow{co\langle k\rangle} \ctmSt{\Sigma'}{P'}$.
First, by means of \ruleref{rule:tr-multicast-context} we split $P$ into two 
subprocesses, one of which contains all threads participating in $k$ 
(those not in $k$ cannot do a transition whose label contains $k$).
Secondly, using recursively \ruleref{rule:tr-multicast-commit} we single
out every thread in $k$. Finally, we apply \ruleref{rule:tr-commit}
provided that every thread is ready to commit, i.e., it is of the form 
$\ctmthreadtr{\textcode{return\;$M$}}{N}{t}{k}$.

Aborting a transaction
works similarly, but it based on vetoes instead of an unanimous vote.
Aborts are triggered by unhandled exceptions raised by some thread, but threads react to this situation
in different ways:
\begin{itemize}[noitemsep]
	\item 
		threads in the same tree of the thread rasing the exception have been forked within the transaction; hence, the root thread is aborted and all other threads in the tree are killed because their creation, as for any transactional side-effect, have to be discarded;
	\item 
		threads in different trees joined the transaction after it was created, due to a merging; hence, these threads just retry their transaction, since aborting
		would require them to handle exceptions raised by ``foreign'' threads.
\end{itemize}

Notice that there are no derivation rules for $\textcode{retry}$, since its meaning is
to inform the scheduler that the execution is stuck; hence the machine has to re-execute the transaction from the beginning
(or a suitable check-point), following a different execution order, if and when possible.

\subsection{Opacity}  
In this section we use the formalisation of \libOTM to prove that it meets the \emph{opacity} criterion.

The opacity correctness criterion for transactional memory \cite{gk:ppopp08}
is an extension of the classical \emph{serialisability property} for databases
with the additional requirement that even non-committed transactions must
access consistent states. Intuitively, this property ensures that:
\begin{enumerate*}[label=\em(\alph{*})]
\item effects of any committed transaction appear performed at a single, indivisible point during the transaction lifetime,
\item effects of any aborted transaction cannot be seen by any other transaction, and
\item transactions always access consistent states of the system.
\end{enumerate*}

In order to formally capture these intuitive requirements let us recall some notions from \cite{gk:ppopp08}.
A \emph{history} is a sequence of \texttt{read}, \texttt{write}, \texttt{commit}, and \texttt{abort} operations\footnote{The definition given in \cite{gk:ppopp08} considers finer-grained events; in particular, \texttt{read} and \texttt{write} operations are formed by \texttt{request}, \texttt{execution}, and \texttt{response} events. However in \textit{loc.~cit.}~the authors restrict to histories where \texttt{request}-\texttt{execution}-\texttt{response} sequences are not interleaved, hence we can consider the simpler \texttt{read}/\texttt{write}s events in the first place.}
ordered according to the time at which they were issued (simultaneous events are arbitrarily ordered) and such that no operation can be issued by a transaction that has already performed a \texttt{commit} or an \texttt{abort}.
A transaction $k$ is said to be in a history $H$ if the latter contains at least one operation issued by $k$.
Any history $H$ defines a \emph{happens-before} partial order $\prec_H$ where
$k \prec_H k'$ iff the transaction $k$ becomes committed or aborted in $H$ before $k'$ issues its first operation.
If $\prec_H$ is total then $H$ is called \emph{sequential}.
For a history $H$, let $\mathit{complete}(H)$ be the set of histories obtained by adding either a commit or an abort for every live transaction in $H$.

We are now able to recall Guerraoui and Kapałka's definition\footnote{%
	The original definition requires the history $H$ to be also ``legal'',
	but this notion is relevant only in presence of non-transactional
	operations which both \libSTM and \libOTM prevent by design.
} of opacity.
\begin{definition}[\hspace{-.1ex}{\hspace{.1ex}\cite[Def.~1]{gk:ppopp08}}]
	A history $H$ is said to be \emph{opaque} if there is a sequential
	history $S$ equivalent to some history in the set $\mathit{complete}(H)$ such that ${\prec_S} \subseteq {\prec_H}$.
\end{definition}

As shown in \cite{gk:ppopp08}, opacity corresponds to the absence of mutual dependencies between live transactions, where a dependency is created whenever a transaction reads an information written by another or depends from its outcome.
\begin{definition}[{Opacity graph \cite[Sec.~5.4]{gk:ppopp08}}]
	For a history $H$ 
	let $\ll$ be a total order on the set $T$ of all transactions in $H$.
	An \emph{opacity graph} for $H$ and $\ll$, $OPG(H,\ll)$, is a bi-coloured directed graph on $T$ such that a vertex is \emph{red} if the corresponding transaction is either running or aborted, it is \emph{black} otherwise, and such that there is an edge from $k$ to $k'$ whenever any of the following holds:
	\begin{enumerate}[label=\it(\alph{*}),noitemsep]
		\item
			$k'$ happens-before $k$;
		\item
			$k$ reads something written by $k'$;
		\item	
			$k'$ reads some location written by $k$ and $k' \ll k$;
		\item
			$k'$ is neither running nor aborted and there are a location
			$r$ and a transaction $k''$ such that $k' \ll k''$, $k'$ writes to $r$, and $k''$ reads $r$ from $k$.
	\end{enumerate}	
	The edge is red if the second case applies otherwise it is black.
	If all edges from red nodes in $OPG(H,\ll)$ are also red then
	the graph is said to be \emph{well-formed}.
\end{definition}

Let $H$ be a history and let $k$ be a transaction appearing in it. A \texttt{read} operation by $k$ is said to be \emph{local} (to $k$) whenever the previous operation by $k$ on the same location was a \texttt{write}. A \texttt{write} operation by $k$ is said to be \emph{local} (to $k$) whenever the next operation by $k$ on the same location is a \texttt{write}.
We denote by $\mathit{nonlocal}(H)$ the longest sub-history of $H$ without any local operations. A history $H$ is said \emph{locally-consistent} if every local \texttt{read} is preceded by a \texttt{write} operation that writes the red value; it is said \emph{consistent} if, additionally, whenever some $k$ reads $v$ from $r$ in $\mathit{nonlocal}(H)$ then
some $k'$ writes $v$ to $r$ in $\mathit{nonlocal}(H)$.

\begin{theorem}[\hspace{-.1ex}{\hspace{.1ex}\cite[Thm.~2]{gk:ppopp08}}]
	A history $H$ is opaque if and only if
	\begin{enumerate*}[label=\it(\alph{*})]
		\item
			$H$ is consistent and
		\item
			there exists a total order $\ll$ on the set of transactions
			in $H$ such that $OPG(\mathit{nonlocal}(H),\ll)$ is well-formed and acyclic.
	\end{enumerate*}
\end{theorem}

In \cite{gk:ppopp08} transactions may encapsulate several threads but cannot be merged. Therefore, in order to study opacity of \libOTM we extend the set of operations considered in loc.~cit.~with explicit merges.
Let $k,k'$ be two running transactions in the given history; when they merge, they share their threads, locations, and effects. From this perspective, $k$ is commit-pending and depends from $k'$ and hence in the opacity graph, $k$ is a red node connected to $k'$ by a red edge. Hence, merges can be equivalently expressed at the history level by sequences like:\\[1ex]
\begin{enumerate*}[label=\it(\arabic{*}),itemsep=-2pt]
	\item new $x$;
	\item $k'$ writes on $x$;
	\item $k$ reads from $x$;
	\item $k$ prepares to commit.
\end{enumerate*}\\[1ex]
These are the only dependencies found in histories generated by \libOTM.
\begin{theorem}\label{th:noloops}
	For $H$ a history describing an execution of a \libOTM program
	and a total order $\ll$, $OPG(\mathit{nonlocal}(H),\ll)$ is a forest of red
	edges where only roots may be white.
\end{theorem}
\begin{proof}
By inspection of the rules it is easy to see that
\begin{enumerate*}[label=\it(\alph{*})]
\item
transactions may access only locations they claimed;
\item
claimed locations are released only on \texttt{commit}s, \texttt{abort}s and retries;
\item
transactions have to merge with any transaction holding a location they need.
\end{enumerate*}
Therefore, at any time there is at most one running transaction issuing operations on a given location, hence~\texttt{read}s and \texttt{write}s do not create edges.
Thus edges are created only during the execution of merges and, by inspecting the above implementation, it easy to see that
\begin{enumerate*}[label=\it(\alph{*}), resume]
\item
any transaction can issue at most one merge;
\item
a transaction issuing a merge is a red node;
\item
the edge created by a merge is red.
\end{enumerate*}
Therefore, transactions form a forest made of red edges where any non-root node is red.
\end{proof}

\begin{corollary}[Opacity]
	\libOTM meets the opacity criterion.
\end{corollary}
\begin{proof}
	A forest formed by red edges whose sources are always red is acyclic and well-formed.
\end{proof}

\section{Conclusions}
\label{sec:conclusions}
In this paper we have presented OTM, a programming model supporting  interactions between composable memory transactions. This model separates isolated transactions from non-isolated ones, still guaranteeing atomicity; the latter can interact by accessing to shared variables. Consistency is ensured by transparently \emph{merging} interacting transactions at runtime.  
We have showed the versatility and simplicity of OTM by implementing some examples which are incompatible with isolation, and we have given a formal semantics for OTM, which allowed us to prove that this model satisfies the opacity criterion.

There are two main directions for future work each posing its own challenges. First, like \libSTM, this model supports nesting (via \textcode{orElse}); however, this feature is currently limited to isolated (sub)transactions. Supporting nesting of open transaction requires additional care in the handling of side-effects: is merging transactions at different level of nesting feasible and meaningful or are we breaking the intuition behind the programming model?
Secondly, an implementation is due in order to validate experimentally the model. A possible approach is to implement \libOTM completely in Haskell on top of \libSTM. This solution does not need any specific support from the Haskell RunTime (HRT) but cannot benefit of the performance gains offered by a deeper integration, thus hindering any fair comparison with existing TM models, like \libSTM. On the other hand, integrating \libOTM with the HRT and the Glasgow Haskell Compiler, akin \libSTM, would be more efficient but also more complex and invasive. 

\looseness=-1
We have presented OTM within Haskell (especially to leverage its type system), but this model is general and can be applied to other languages. A possible future work is to port this model to an imperative object oriented language, such as Java or C++; however, like other TM implementations, we expect that this extension will require some changes in the compiler and/or the runtime. 

This work builds on the ideas in \cite{mpt:coord15} where we described an abstract calculus with shared memory and open transactions. In \textit{loc.~cit.}~we showed how this model is expressive enough to represent  $TCCS^m$ \cite{ksh:fossacs2014}, a variant of the Calculus of Communicating Systems with transactional synchronization.
Being based on CCS, communication in $TCCS^m$ is synchronous;
however, nowadays asynchronous models play an important r\^ole (see
e.g.~actors, event-driven programming, etc.).  It may be interesting
to generalize the discussion so as to consider also this case, e.g.~by
defining an actor-based calculus with open transactions.  Such a
calculus can be quite useful also for modelling speculative reasoning
for cooperating systems \cite{ma2010:speculative, mmp:dais14, mmp:eceast2014, mpm:gcm14w, mp:memo14}.
A local version of actor-based open transactions can be
implemented in \libOTM using lock-free data
structures (e.g., message queues) in shared transactional memory. 

\paragraph{Acknowledgements}
We thank Nicola Gigante and Valentino Picotti for their valuable feedback about the \libOTM programming model.

\end{document}